\RequirePackage{plautopatch}
\documentclass[uplatex,divpdfmx,a4paper,11pt]{article}
\usepackage{enumerate}
\usepackage{float}
\usepackage{latexsym,amsmath,mathrsfs,bm,amssymb,color,amsthm,amsbsy,amscd}
\theoremstyle{plain}
\newtheorem{theorem}{Theorem}[section]
\newtheorem{proposition}[theorem]{Proposition}
\newtheorem{lemma}[theorem]{Lemma}
\newtheorem{corollary}[theorem]{Corollary}

\theoremstyle{definition}
\newtheorem{definition}[theorem]{Definition}
\newtheorem{example}[theorem]{Example}
\newtheorem{remark}[theorem]{Remark}
\newtheorem{assumption}[theorem]{Assumption}

\makeatletter
\@addtoreset{equation}{section}
\makeatother

\allowdisplaybreaks[3]

\def\bbbone{{\mathchoice {\rm 1\mskip-4mu l} {\rm 1\mskip-4mu l}
		{\rm 1\mskip-4.5mu l} {\rm 1\mskip-5mu l}}}
\def\one{\bbbone}

\begin{document}
\title{\sc Self-adjointness of unbounded time operators}
\author{Fumio Hiroshima\footnote{Faculty of Mathematics, Kyushu University} and Noriaki Teranishi\footnote{Faculty of Science, Department of Mathematics, Hokkaido University}}
\date{\today}
\maketitle
\begin{abstract}
Time operators for an abstract semi-bounded self-adjoint operator $H$ 
with purely discrete spectrum is considered. 
The existence of a bounded self-adjoint time operator $T$ for $H$ is known as Galapon time operator. 
In this paper, a self-adjoint but \emph{unbounded} time operator $T$ of $H$ is constructed. 
\end{abstract}

\section{Introduction}
We give the definition of time operators and conjugate operators. 

\begin{definition} 
Let $H$ be a self-adjoint operator on a Hilbert space $\mathcal H$ and $T$ an operator on $\mathcal H$.
If $H$ and $T$ satisfy the canonical commutation relation 
\begin{align*}
[H,T]=HT-TH=-i\one
\end{align*}
on $D_{H,T}\subset \mathrm{D} (HT)\cap \mathrm{D} (TH)$ but 
$D_{H,T}\neq\{0\}$, 
then $T$ is called a conjugate operator of $H$ and $D_{H,T}$ is called a CCR-domain. 
Here $\mathrm{D} (A)$ is the domain of the operator $A$.
If $T$ is a symmetric operator on $\mathcal H$, then $T$ is called a time operator of $H$. 
\end{definition}

Time operators and/or conjugate operators for $H$ are in general not unique. In the series of papers \cite{HT24a,HT24b} 
we construct time operators and/or conjugate operators for $1$D-harmonic oscillator. 
It is known that the so-called Galapon time operator $T_{\mathrm G}$ for the $1$D-harmonic oscillator is bounded self-adjoint operator and the CCR-domain is dense.

We introduce Galapon operator. 
A self-adjoint operaotr $H$ considering the Galapon operator usually imposes the following conditions.

\begin{assumption}\label{ass}
An operator $H$ on a separable Hilbert space $\mathcal H$ is positive, unbounded and self-adjoint. 
The spectrum $\sigma(H)$ of $H$ consists of only simple eigenvalues and $H^{-1}$ is Hilbert-Schmidt. 
\end{assumption}
Let $H$ be an operator which satisfies Assumtion \ref{ass}, $e_n$ an eigenvector of $H$ for an eigenvalue $E_n$ for $n\in\mathbb N$. 
Note that \[\sum_{n=0}^\infty \frac{1}{E_n^2}<\infty.\] 
Define the Galapon time operator $T_{\mathrm G}$ associated with the operator $H$ by 
\begin{align*}
&\mathrm{D} (T_{\mathrm G})={\rm LH}\{e_n\in\mathcal H \mid n\in\mathbb N \},\\
&T_{\mathrm G}\varphi=i\sum_{n=0}^\infty\left( \sum_{m\neq n}\frac{(e_m, \varphi)}{E_n-E_m}\right) e_n,\quad \varphi\in\mathrm{D} (T).
\end{align*}
Here, for a subset $\mathcal A$ of $\mathcal H$, $\mathrm{LH}\mathcal A$ means the linear hull of $\mathcal A$.

\begin{proposition}[\cite{gal02a,AM08}]
\label{gal}
Suppose that $H$ satisfies Assumption \ref{ass}. 
Then $T_{\mathrm G}$ is a densely defined time operator of $H$ with a CCR-domain 
${\rm LH}\{e_n-e_m\mid n,m\in\mathbb N \}$. 
\end{proposition}
It is established in \cite[Theorem 4.5]{AM08}
that 
if 
\begin{align}\label{enm}
E_n-E_m\geq C(n^\lambda-m^\lambda)
\end{align}
for some constants $C>0$ and $\lambda>1$, 
then $T_{\mathrm G}$ is a bounded time operator of $H$, and hence it is self-adjoint. 
In particular, if
\[E_n=a n^\lambda+b,\quad n\in \mathbb N,\quad \lambda>1\]
with some constants $a,b>0$, then 
$T_{\mathrm G}$ is a bounded self-adjoint time operator. 

Let 
\[E_n=a n^\lambda+b,\quad n\in\mathbb N,\quad 1/2<\lambda<1\]
with constants $a,b>0$. 
Then it can be shown that 
$T_{\mathrm G}$ is unbounded time operator. 
The self-adjointness of $T_{\mathrm G}$ is however unknown. 
It is also pointed out in \cite[Remark 4.7]{ara20} that 
no examples of {\it unbounded} self-adjoint time operator of the form $T_{\mathrm G}$ is constructed.

The purpose of this paper is to construct 
unbounded self-adjoint time operators 
for some abstract self-adjoint operator $H$. 
We organize this paper as follows. 
In Section 2 we unitarily transform $T_{\mathrm G}$ to an operator 
$T_f$ on $\ell^2(\mathbb N)$. 
Section 3 is devoted to constructing 
unbounded self-adjoint time operators. 
The main results are stated in 
Theorems \ref{sa_T3} and \ref{sa_T5}. 

\section{Galapon time operator on $\bm{\ell^2(\mathbb N)}$}
In this paper, the investigation of time operators is carried out on 
$\ell^2(\mathbb N)$ instead of $\mathcal H$. 
Then we shall show first of all 
that $T_{\mathrm G}$ is unitarily equivalent to an operator $T_ f$ on 
$\ell^2(\mathbb N)$. 
Denote by $\ell^2(\mathbb N)$ the set of square summable functions on $\mathbb N $ and 
let $\xi_n\in\ell^2(\mathbb N)$ be the function on $\mathbb N$ defined by
\begin{align*}
\xi_n(m)=\delta_{nm},\quad m\in\mathbb N,
\end{align*}
where $\delta_{nm}$ is the Kronecker delta function. 
We write $\ell^2_\mathrm{fin}(\mathbb N)$ for the set of $\varphi\in\ell^2(\mathbb N)$ which has a finite support, i.e., there exists $m\in\mathbb N $ and $(c_n)_{n=0}^m\in\mathbb C^{m+1}$ such that $\varphi$ is represented by $\sum_{n=0}^m c_n\xi_n$. 
Note that $\ell^2_\mathrm{fin}(\mathbb N)$ is dense in $\ell^2(\mathbb N)$. 
Let $L$ be the left shift operator on $\ell^2(\mathbb N)$ and $L^\ast$ the adjoint operator of $L$;
\begin{align*}
L^\ast\xi_n=\xi_{n+1},\quad n\in\mathbb N.
\end{align*}
Let $N$ be the number operator on $\ell^2(\mathbb N)$. Then 
$N\xi_n=n\xi_n$ for $n\in\mathbb N$. 
It is well known that $N$ is a self-adjoint operator, $\ell^2_\mathrm{fin}(\mathbb N)$ is a core for $N$ and $N$ satisfies 
$[N,L]=-L$ and $[N,L^\ast]=L^\ast$ on $\ell^2_\mathrm{fin}(\mathbb N)$. 

We introduce notations $\mathcal K$ and $\mathcal K^{-}$. 
\begin{definition}[$\mathcal K$ and $\mathcal K^-$]\label{A1}
Let us denote by $\mathcal K$ the set of all real valued funtions on $\mathbb N$ which satisfy the following conditions:
\begin{enumerate}
\item [\rm (1)] $ f(0)>0$,
\item [\rm (2)] $ f(n)< f(n+1)$ for all $n\in\mathbb N $.
\end{enumerate}
Write $\mathcal K^{-}=\bigl\{f\in\mathcal K\ \big|\ 1/f\in\ell^2(\mathbb N)\bigr\}$.
\end{definition}

To define $T_ f$ for $ f\in\mathcal K$, we 
set 
\begin{equation*}
\Delta_k(f,n)=f(n+k)-f(n).
\end{equation*}
\begin{lemma}\label{dom_fin}
Let $f\in \mathcal K$. Then $\ell^2_\mathrm{fin}(\mathbb N)\subset\mathrm{D} (\Delta_k(f,N)^{-1})$ for all natural number $k\geq 1$. 
\end{lemma}

\begin{proof} Since $f$ is strictly increasing, $\Delta_k(f,N)$ is injective. 
Clearly $\ell^2_\mathrm{fin}(\mathbb N)\subset\mathrm{D} (\Delta_k(f,N))$ and $\xi_n$ is an eigenvector of $\Delta_k(f,N)$:
\begin{align*}
\Delta_k(f,N)\xi_n=\Delta_k(f,n)\xi_n.
\end{align*}
This implies that $\ell^2_\mathrm{fin}(\mathbb N)\subset\mathrm{D} (\Delta_k(f,N)^{-1})$.
\end{proof}

\begin{remark}
Note that $\inf_{n\in\mathbb N }\Delta_k(f,n)>0$ if and only if 
$\Delta_k(f,N)^{-1}$ is a bounded operator.
\end{remark}

\begin{definition}
Let $f\in \mathcal K$. 
We define operators $T_{f,m}$ and $T_ f$ on 
$\ell^2(\mathbb N)$ by
\begin{align*}
&T_{f,m}=i\sum_{k=1}^m \left( L^\ast{}^k\Delta_k(f,N)^{-1}- \Delta_k(f,N)^{-1}L^k\right),\\
&\mathrm{D} (T_ f)=\left\{\varphi\in\bigcap_{m\geq1}\mathrm{D} \left(T_{f,m}\right)\ \middle|\ \lim_{m\to\infty}T_{f,m}\varphi\text{ exists in }\ell^2(\mathbb N)\right\},\\
&T_ f\varphi=\lim_{m\to\infty}T_{f,m}\varphi,\quad \varphi\in\mathrm{D} (T_ f).
\end{align*}
\end{definition}

\begin{lemma}\label{dom_Tf}
Suppose that $f\in \mathcal K^{-}$. 
Then $\ell^2_\mathrm{fin}(\mathbb N)\subset\mathrm{D} \left(T_ f\right)$.
\end{lemma}

\begin{proof}
It is sufficient to show that $\lim_{m\to\infty}T_{f,m}\xi_n$ exists for all $n\in\mathbb N $. 
For any $n\leq m_1\leq m_2$, 
\begin{align*}
\left\|\left(T_{f, m_2}-T_{f, m_1}\right)\xi_n\right\|^2
&=\left\|\sum_{k=m_1+1}^{m_2}\left(L^\ast{}^k\Delta_k\left(f, N\right)^{-1}-\Delta_k\left(f, N\right)^{-1}L^k\right)\xi_n\right\|^2\\
&=\sum_{k=m_1+1}^{m_2} \frac{1}{(f(n+k)- f(n))^2}\\
&\leq\left(1-\frac{f(n)}{f(n+1)}\right)^{-2}\sum_{k=m_1+1}^{m_2}\frac{1}{f(n+k)^2}\to 0
\end{align*}
as $m_1,m_2\to\infty$. 
Hence $\{T_{f,m}\xi_n\}_{m\in\mathbb N }$ is a Cauchy sequence. 
Therefore $\lim_{m\to\infty}T_{f,m}\xi_n$ exists and $\xi_n\in\mathrm{D} (T_ f)$.
\end{proof}
From Lemma \ref{dom_Tf}, we see that $T_ f$ is a densely defined symmetric operator. 
The relationship between $T_ f$ and $T_{\mathrm G}$ is given by the following theorem.
\begin{theorem}\label{UT}
Suppose that $H$ satisfies Assumption \ref{ass}. 
Then there exists a unitary operator $U\colon\mathcal H\to\ell^2(\mathbb N)$ 
and a function $ f\in\mathcal K^{-}$ such that 
$ f(N)=UHU^\ast$ and $T_ f$ is unitary equivalent to $T_{\mathrm G}$ on $\ell^2_\mathrm{fin}(\mathbb N)$, i.e., 
\[UT_{\mathrm G} U^\ast=T_ f\quad \mbox{ on } 
\ell^2_\mathrm{fin}(\mathbb N).\]
\end{theorem}

\begin{proof} 
Let 
$ f\colon\mathbb N \to\mathbb R$ be a function such that $ f(n)=E_n$. 
 Then 
$ f\in \mathcal K^-$. 
Let $U\colon \mathcal H\to\ell^2(\mathbb N)$ be the unitary operator defined by
$Ue_n=\xi_n$ for any $n\in\mathbb N $. 
For arbitrary $\varphi\in\mathrm{D} (T)$, we see that 
\begin{align*}
UT_{\mathrm G}\varphi
&=i\sum_{n=0}^\infty\left( \sum_{m<n}\frac{\left( \xi_m, U\varphi\right) }{E_n-E_m}+\sum_{m>n}\frac{\left( \xi_m, U\varphi\right) }{E_n-E_m}\right) \xi_n\\
&=i\sum_{n=0}^\infty\left( \sum_{m<n}\frac{\left( L^{n-m}\xi_n, \varphi\right) }{E_n-E_m}+\sum_{m>n}\frac{\left( L^\ast{}^{m-n}\xi_n, U\varphi\right) }{E_n-E_m}\right) \xi_n\\
&=i\sum_{n=0}^\infty\left( \sum_{k=1}^\infty\frac{\left( L^k \xi_n, U\varphi\right) }{E_n-E_{n-k}}-\sum_{k=1}^\infty\frac{\left( L^\ast{}^k \xi_n, U\varphi\right) }{E_{n+k}-E_n}\right) \xi_n.
\end{align*}
Since $ f(N)\xi_n=E_n\xi_n$, it follows that
\begin{align*}
(E_{n}-E_{n-k})^{-1}L^k\xi_n=\Delta_k(f,N)^{-1}L^k\xi_n
\end{align*}
and
\begin{align*}
(E_{n+k}-E_n)^{-1}\xi_n=\Delta_k(f,N)^{-1}\xi_n.
\end{align*}
From Lemma \ref{dom_Tf}, we see that $\xi_n, U\varphi\in\mathrm{D} (T_ f)$. 
Thus 
\begin{align*}
UT_{\mathrm G}\varphi &=i\sum_{n=0}^\infty\left( 
\sum_{k=1}^\infty\left( 
\Delta_k(f,N)^{-1}L^k-L^\ast{}^k\Delta_k(f,N)^{-1}\right) \xi_n, U\varphi\right)\xi_n\\
&=\sum_{n=0}^\infty
\left( \xi_n, i\sum_{k=1}^\infty\left( L^\ast{}^k\Delta_k(f,N)^{-1}- \Delta_k(f,N)^{-1}L^k\right) U\varphi\right) \xi_n.
\end{align*}
This shows that 
$UT_{\mathrm G}\varphi=T_ f U\varphi$ for any $\varphi\in\mathrm{D} (T_{\mathrm G})=U^\ast\ell^2_\mathrm{fin}(\mathbb N)$. 
Then the theorem is proven. 
\end{proof}
\begin{corollary}\label{UTc}
For all $f\in\mathcal K^{-}$, 
the operator $T_f$ is a time operator of $f(N)$ with a CCR domain 
$(\one-L^\ast)\ell^2_\mathrm{fin}(\mathbb N)$.
\end{corollary}

\begin{proof} 
By the definition of $\mathcal K^-$, 
the operator $f(N)$ is a positive and unbounded self-adjoint operator, 
and 
$\sigma(f(N))$ consists of only simple eigenvalues and $f(N)^{-1}$ is Hilbert-Schmidt. 
Thus $T_ f$ is a time operator of $ f(N)$ by Theorem \ref{UT}.
\end{proof}
By Theorem \ref{UT} and Corollary \ref{UTc}, 
the set $\{T_f\mid f\in\mathcal K^{-}\}$ includes Galapon time operators $T_{\mathrm G}$. So in what follows we consider 
time operator $T_f$.

\section{Self-adjointness of time operators}
\subsection{Bounded cases}
Let us recall the case where the operator $T_f$ is bounded.  
\begin{lemma}
\label{bd_T}
Let $f\in \mathcal K$. 
Suppose that $0\not\in\sigma(\Delta_k(f,N))$ for all $k\geq1$ and 
\[\sum_{k\geq1}\left\|\Delta_k(f,N)^{-1}\right\|<\infty.\] 
Then the operator $T_ f$ is bounded.
In particular, 
$T_ f$ is a self-adjoint operator. 
\end{lemma}

\begin{proof}
 For any $\varphi\in\ell^2(\mathbb N)$ and $1\leq m_1<m_2$,
\begin{align*}
\left\|\left(T_{f,m_2}-T_{f,m_1}\right)\varphi\right\|&\leq\sum_{k=m_1+1}^{m_2}\left(\left\|L^\ast{}^k\Delta_k(f,N)^{-1}\right\|+\left\|\Delta_k(f,N)^{-1}L^k\right\|\right)\|\varphi\|\\
&\leq2\|\varphi\|\!\!\!\!\sum_{k=m_1+1}^{m_2}\!\!\! \left\|\Delta_k(f,N)^{-1}\right\|.
\end{align*}
This shows that $\{T_{f,m}\varphi\}_{m\in\mathbb N }$ is a Cauchy sequence. Therefore $\mathrm{D} (T_ f)=\ell^2(\mathbb N)$ and $T_ f$ is bounded.
\end{proof}
A similar result to Lemma \ref{bd_T} is obtained 
in \cite[Theorem 4.5]{AM08}. 
\begin{example}
Let $\lambda>1$ and $ f(x)=x^\lambda+1$. 
Then $f\in \mathcal K^-$. 
Since $\Delta_k(f,n)\geq\Delta_k(f, 0)=k^\lambda$, we have 
\begin{align*}
\sum_{k\geq1}\left\|\Delta_k(f,N)^{-1}\right\|\leq \sum_{k\geq1}k^{-\lambda}<\infty.
\end{align*}
Therefore $T_ f$ is bounded self-adjoint time operator of $ f(N)$.
\end{example}

\subsection{Unbounded cases}
Next proposition is a sufficient condition for $T_ f$ to be 
unbounded. 
\begin{proposition}\label{unb}
Suppose that $f\in \mathcal K^-$ and 
$0\in\sigma(\Delta_1(f,N))$. Then 
$T_ f$ is unbounded.
\end{proposition}

\begin{proof}
	 See \cite[Theorem 5.1]{AM08}.
\end{proof}
Let $f\colon\mathop{\mathrm{dom}}(f)\to\mathbb C$. 
In this paper, we denote by 
$f^2$ the function $f^2\colon\mathop{\mathrm{dom}}(f)\to\mathbb C$, 
$f^2(x)=f(x)^2$ for each $x\in\mathop{\mathrm{dom}}(f)$. 
In what follows we consider operators of the form $ f(N)T_{f^2}+T_{f^2} f(N)$. 

\begin{lemma}\label{D_fT}
	Let $f\in \mathcal K^-$. 
	Then 
	$\ell^2_\mathrm{fin}(\mathbb N)\subset\mathrm{D} \left(f(N)T_{f^2}\right)$ 
	and 
	\begin{equation*}
		\lim_{m\to\infty} f(N)T_{f^2, m}\xi_n=f(N)T_{f^2}\xi_n
	\end{equation*}
	for all $n\in\mathbb N$.
\end{lemma}

\begin{proof} Similarly to the proof of Lemma \ref{dom_Tf}, 
	for any $n\leq m_1\leq m_2$, we have 
	\begin{align*}
		\left\| f(N)\left(T_{f^2, m_2}-T_{f^2, m_1}\right)\xi_n\right\|^2
		&=\sum_{k=m_1+1}^{m_2} \frac{f(n+k)^2}{(f(n+k)^2- f(n)^2)^2}\\
		&\leq\left(1-\frac{f(n)}{f(n+1)}\right)^{-2}\sum_{k=m_1+1}^{m_2}\frac{1}{f(n+k)^2}.
	\end{align*}
	Therefore $\lim_{m\to\infty} f(N)T_{f^2, m}\xi_n$ exists.
	Since $f(N)$ is a closed operator, we obtain desired conclusion. 
\end{proof}

The next lemma shows that $T_ f$ is identical to $ f(N)T_{f^2}+T_{f^2} f(N)$ 
on $\ell^2_{\rm fin}(\mathbb N)$. 

\begin{lemma}\label{T_p1}
Let $f\in \mathcal K^-$. 
Then 
\begin{align}\label{good}
 f(N)T_{f^2}+T_{f^2} f(N)= T_ f
\end{align}
on $\ell^2_\mathrm{fin}(\mathbb N)$ and 
\begin{align*}
[ f(N), f(N)T_{f^2}+T_{f^2} f(N)]=-i\one
\end{align*}
on $(\one-L^\ast)\ell^2_\mathrm{fin}(\mathbb N)$.
\end{lemma}

\begin{proof} 
From Lemma \ref{D_fT}, for any $\varphi\in\ell^2_{\rm fin}(\mathbb N)$, 
\begin{equation*}
\left(f(N)T_{f^2}+T_{f^2} f(N)\right)\varphi
=\lim_{m\to\infty}\left(f(N)T_{f^2,m}+T_{f^2,m} f(N)\right)\varphi.
\end{equation*}
For each $m$, we obtain 
\begin{align*}
&\left(f(N)T_{f^2,m}+T_{f^2,m} f(N)\right)\varphi\\
&=i\sum_{k=1}^m \left(L^\ast{}^k\left(f(N+k)+ f(N)\right)\Delta_k\left(f^2, N\right)^{-1}-\Delta_k\left(f^2, N\right)^{-1}\left(f(N+k)+ f(N)\right)L^k\right)\varphi\\
&=i\sum_{k=1}^m \left(L^\ast{}^k\Delta_k(f,N)^{-1}-\Delta_k(f,N)^{-1}L^k\right)\varphi\\
&=T_{f,m}\varphi.
\end{align*}
Hence we see that $\varphi\in\mathrm D(T_f)$ 
and $ f(N)T_{f^2}+T_{f^2} f(N)= T_ f$ on $\ell^2_{\rm fin}(\mathbb N)$. 

Since $T_f$ is a time operator of $f(N)$ with a CCR-domain $(\one-L^\ast)\ell^2_\mathrm{fin}(\mathbb N)$, 
\begin{equation*}
\left[f(N), f(N)T_{f^2}+T_{f^2}f(N)\right]=[f(N), T_f]=-i\one 
\end{equation*}
on $(\one-L^\ast)\ell^2_\mathrm{fin}(\mathbb N)$.
\end{proof}

Intuitively it may be hard to show the (essentially) self-adjointness of 
$ f(N)T_{f^2}+T_{f^2} f(N)$ or $T_ f$ themselves, 
since 
operators $ f(N)T_{f^2}+T_{f^2} f(N)$ and $T_ f$
are unbounded both from above and below, 
and a CCR-domain $(\one-L^\ast)\ell^2_\mathrm{fin}(\mathbb N)$ is not a core of $ f(N)$. 
So we add extra term $ f(N)^\beta$ to 
$ f(N)T_{f^2}+T_{f^2} f(N)$. 
Note that $[N, f(N)^\beta]\subset 0$. 
Hence we consider 
$ f(N)T_{f^2}+T_{f^2} f(N)+\gamma f(N)^\beta$ instead of 
$ f(N)T_{f^2}+T_{f^2} f(N)$ and 
show that 
$ f(N)T_{f^2}+T_{f^2} f(N)+\gamma f(N)^\beta$ 
is self-adjoint by the fact that 
$ f(N)T_{f^2}+T_{f^2} f(N)$ is relatively small with respect to 
$ f(N)^\beta$. 

We introduce classes $\mathcal M(\beta)$ and $\mathcal M_\mathrm{s}(\beta)$ of functions on $\mathbb N$. 
\begin{definition}[$\mathcal M(\beta)$ and $\mathcal M_\mathrm{s}(\beta)$]\label{A2}
Let $\beta\geq0$. Denote by $\mathcal M(\beta)$ the set of all functions $f\in\mathcal K^-$ such that 
there exist functions
$g\colon\mathbb N \to(0,\infty)$ and 
$h\in\ell^1(\mathbb N _{\geq1}, \mathbb R )$ satisfying the following conditions:
\begin{enumerate}[{\rm (1)}]
\item $ f^2/g\in\ell^1(\mathbb N)$,
\item for any $n\in\mathbb N $ and $k\geq 1$,
\begin{equation}\label{ac1}
\frac{g(n)}{\left(f(n)^{\beta}\Delta_k(f^2, n)\right)^{2}}\leq h(k)^2.
\end{equation}
\end{enumerate} 
Write $\mathcal M_\mathrm{s}(\beta)$ the set of functions $f\in \mathcal M(\beta)$ that, for the above function $g$, 
there exists a constant $C>0$ such that
\begin{equation}\label{ac2}
\sup_{n\in\mathbb N}
\sum_{k=1}^n\frac{g(n)}{\left( f(n-k)^{\beta}\left(f(n)^2- f(n-k)^2\right)\right)^2}<C.
\end{equation}
\end{definition}

\begin{lemma}\label{bd_T2}
Let $f\in \mathcal M(1)$. 
Then $T_{f^2}$ is bounded.
\end{lemma}

\begin{proof} 
By (1) of Definition \ref{A2}, $\sup_{n\in\mathbb N } f(n)^2/g(n)$ is finite. 
From \eqref{ac1}, we have
\begin{align*}
\left\|\Delta_k\!\left(f^2, N\right)^{-1}\right\|&\leq\sup_{n\in\mathbb N }\Delta_k\!\left(f^2, n\right)^{-1}
\leq\sup_{n\in\mathbb N }\left(f(n)^2/g(n)\right)^{1/2}h(k).
\end{align*}
Since $h\in\ell^1(\mathbb N _{\geq1},\mathbb R )$, by Lemma \ref{bd_T}, we see that $T_{f^2}$ is bounded .
\end{proof}
\begin{lemma}\label{rb0}
Let $f\in \mathcal M_\mathrm{s}(\beta)$
Then $ f(N)T_{f^2}$ is relatively bounded with respect to 
$ f(N)^\beta$, i.e., there exists some constant $a>0$ such that for all 
$\varphi\in\mathrm{D}(f(N)^\beta)$
\begin{align*}\label{rb1}
\left\|\overline{f(N)T_{f^2}}\varphi\right\|\leq a\left\| f(N)^\beta\varphi\right\|.
\end{align*}
Here $\overline{A}$ denotes the closure of the linear operator $A$.
\end{lemma}

\begin{proof}
It is sufficient to show that $\overline{f(N)T_{f^2}} f(N)^{-\beta}$ is bounded. 
For any $\varphi\in\ell^2_\mathrm{fin}(\mathbb N)$, 
\begin{align*}
\left\| f(N)T_{f^2} f(N)^{-\beta}\varphi\right\|^2&=\sum_{n=0}^\infty\left(\xi_n, f(N)T_{f^2} f(N)^{-\beta}\varphi\right)^2\\
&\leq\|\varphi\|^2\sum_{n=0}^\infty\left\| f(N)^{-\beta}T_{f^2} f(N)\xi_n\right\|^2\\
&=\|\varphi\|^2\sum_{n=0}^\infty f(n)^2\left\| f(N)^{-\beta}T_{f^2}\xi_n\right\|^2.
\end{align*}
For all $n\in\mathbb N$, we see that
\begin{align*}
\left\| f(N)^{-\beta}T_{f^2}\xi_n\right\|^2
&=\left\|\sum_{k\geq1} f(N)^{-\beta}\left(L^\ast{}^k\Delta_k(f^2,N)^{-1}-\Delta_k(f^2,N)^{-1}L^k\right)\xi_n\right\|^2\\
&=\sum_{k\geq1}\frac{1}{f(n+k)^{2\beta}\Delta_k\left(f^2, n\right)^2}+\sum_{k=1}^n\frac{1}{f(n-k)^{2\beta}(f(n)^2- f(n-k)^2)^2}\\
&\leq \frac{1}{g(n)}\left(\sum_{k\geq1}h(k)^2+C\right).
\end{align*}
Thus we have
\begin{align*}
\left\|f(N)T_{f^2} f(N)^{-\beta}\varphi\right\|^2\leq \|\varphi\|^2
\left(\|h\|_{\ell^2}^2+C\right)\sum_{n=0}^\infty f(n)^2g(n)^{-1}.
\end{align*}
Therefore $\overline{f(N)T_{f^2}}f(N)^{-\beta}$ is bounded.
\end{proof}

\begin{remark}\label{rb_T}
	Let $\beta\geq 1$, $f\in \mathcal M_\mathrm{s}(\beta)$ and $T_f^2$ be bounded. 
	Since the operator $T_f$ is equal to $f(N)T_{f^2}+T_{f^2}f(N)$ on $\ell^2_\mathrm{fin}(\mathbb N)$ by Lemma \ref{T_p1}, 
	it is easy to see that $\overline{T_f}$ is also relatively bounded with respect to $f(N)^\beta$, i.e., 
	there exists a constant $a>0$ such that for all $\varphi\in\mathrm{D}(f(N)^\beta)$
	\begin{equation*}
		\left\|\overline{T_{f}}\varphi\right\|\leq a\left\|f(N)^\beta\varphi\right\|.
	\end{equation*}
\end{remark}

\begin{proposition}\label{sa_T2}
Let $f\in \mathcal M_\mathrm{s}(1)$. 
Then {\rm (1)} and {\rm (2)} follow.
\begin{enumerate}[{\rm (1)}]
	\item $ f(N)T_{f^2}+T_{f^2} f(N)+\gamma f(N)$ is self-adjoint 
		if $|\gamma|\gg1$.
	\item $f(N)T_{f^2}+T_{f^2} f(N)+\gamma f(N)^2$ is self-adjoint 
		if $\gamma\in\mathbb R \setminus\{0\}$.
\end{enumerate}
\end{proposition}

\begin{proof} 
(1) From Lemma \ref{rb0}, there exists a relative bound $a$ of $ f(N)T_{f^2}+T_{f^2} f(N)$ with respect to $ f(N)$. 
Thus if $|\gamma|>a$, by the Kato-Rellich theorem, the operator is self-adjoint. 

(2) 
From Lemma \ref{rb0}, we see that $ f(N)T_{f^2}+T_{f^2} f(N)$ is relatively bounded with respect to $ f(N)$. 
Since $ f(N)$ is infinitesimally small with respect to $ f(N)^2$, by the Kato-Relich theorem, $\gamma f(N)^2+ f(N)T_{f^2}+T_{f^2} f(N)$ is a self-adjoint operator.
\end{proof}

We are in the position to state the main theorem in this paper. 
\begin{theorem}\label{sa_T3}
Let $f\in \mathcal M_\mathrm{s}(1)$. 
Then 
$\overline{T_f}+\gamma f(N)^2$ is a self-adjoint time operator of $f(N)$ with a dense CCR-domain for all $\gamma\in\mathbb R \setminus\{0\}$. 
\end{theorem}

\begin{proof}
	 From \eqref{good}, Lemma \ref{bd_T2}, Lemma \ref{rb0} and Remark \ref{rb_T}, we see that $\overline{T_f}+\gamma f(N)^2$ is self-adjoint for all $\gamma\in\mathbb R \setminus\{0\}$ and 
$[ f(N), T_ f+\gamma f(N)^2]=-i\one$ 
on $(\one-L^\ast)\ell^2_\mathrm{fin}(\mathbb N)$.
\end{proof}

\begin{example}\label{3/4}
Let 
$ f(x)=x^\lambda+1$ for
$\lambda\in(3/4,1)$. 
We show that $f\in \mathcal M_\mathrm{s}(1)$. 
Firstly, it is immediate that $f\in\mathcal K^-$. 
Let $\alpha\in(1+2\lambda, 6\lambda-2)$, $g(x)=x^\alpha+1$ and $\delta=6\lambda-2-\alpha$. 
Then the condition (1) of $\mathcal M(1)$ is satisfied.

Secondly, 
by the mean value theorem, we have
\begin{align*}
 f(n+k)- f(n)\geq \frac{\lambda k}{(n+k)^{1-\lambda}}.
\end{align*}
Thus we see that 
\begin{align*}
\frac{g(n)}{f(n)^{2}\Delta_k\left(f^2, n\right)^{2}}&=\frac{n^\alpha+1}{f(n)^2(f(n+k)^2- f(n)^2)^2}
\leq\frac{(n^\alpha+1)(n+k)^{2(1-\lambda)}}{\lambda^2(n^\lambda+1)^2(n+k)^{2\lambda}k^2}
\leq\frac{2}{\lambda^2k^{2+\delta}}.
\end{align*}
Thus the condition (2) of $\mathcal M(1)$ is satisfied and $f\in\mathcal M(1)$. 

Finally we see that 
\begin{align*}
&\lim_{n\to\infty}\sum_{k=1}^n\frac{g(n)}{f(n-k)^2\left(f(n)^2- f(n-k)^2\right)^2}\\
&=\lim_{n\to\infty}\left(
\sum_{k=1}^{[n/2]}\frac{g(n)}{f(n-k)^2\left(f(n)^2- f(n-k)^2\right)^2}
+\sum_{k=[n/2]+1}^n\frac{g(n)}{f(n-k)^2\left(f(n)^2- f(n-k)^2\right)^2}\right)\\
&\leq\lim_{n\to\infty}\frac{4^\lambda (n^\alpha+1) n^{2(1-\lambda)}}{\lambda^2(n^\lambda+1)^2n^{2\lambda}}\left(\sum_{k=1}^{[n/2]}\frac{1}{k^2}
+\frac{4^{1-\lambda}}{n^{2(1-\lambda)}}\sum_{k=[n/2]+1}^n\frac{1}{f(n-k)^2}\right)<\infty.
\end{align*}
Then the condition \eqref{ac2} is satisfied and $f\in\mathcal M_\mathrm{s}(1)$. 

We see that $T_f$ is unbounded by Proposition \ref{unb}
and that, from Theorem \ref{sa_T3}, $f(N)$ has a self-adjoint time operator with a dense CCR-domain.
\end{example}

Similarly to Proposition \ref{sa_T2} and Theorem \ref{sa_T3}, 
we can show the following statements. We omit proofs.
\begin{proposition}\label{sa_T5}
Let $f\in \mathcal M_\mathrm{s}(2)$ and $T_{f^2}$ is bounded. 
Then 
$f(N)T_{f^2}+T_{f^2}f(N)+\gamma f(N)^2$ is a self-adjoint time operator of $f(N)$ if $|\gamma|\gg1$. 
\end{proposition}
\begin{theorem}\label{sa_T6}
Let $f\in \mathcal M_\mathrm{s}(2)$ and $T_{f^2}$ is bounded. 
Then 
$\overline{T_f}+\gamma f(N)^3$ is a self-adjoint time operator of $f(N)$ with a dense CCR-domain for all $\gamma\in\mathbb R \setminus\{0\}$. 
\end{theorem}

\begin{example}
Let $ f(x)=x^\lambda+1$ for $\lambda\in(1/2,1)$. Then 
$T_{f^2}$ is bounded by Lemma \ref{bd_T}. 
Similar to Example \ref{3/4}, we can see that $f\in \mathcal M_\mathrm{s}(2)$. 
Therefore $ f(N)$ has an unbounded self-adjoint time operator with a dense CCR-domain by Theorem \ref{sa_T6}.
\end{example}

\subsection{Extensions}
In this section we consider the case that $f^2\in\mathcal K^-$ but not necessarily $f\in \mathcal K^-$. 
\begin{proposition}\label{fr}
Let $f\in \mathcal K$. 
If $\mathrm{D} (f(N)T_{f^2})\cap\mathrm{D} (f(N)^2)$ is dense and $T_{f^2}$ is bounded, 
then $ f(N)T_{f^2}+T_{f^2} f(N)+\gamma f(N)^2$ of $ f(N)$ has a self-adjoint extension for all $\gamma\geq1$. 
\end{proposition}

\begin{proof}
	From 
\begin{align*}
 f(N)T_{f^2}+T_{f^2} f(N)+\gamma f(N)^2\subset\left(f(N)+T_{f^2}\right)^2+(\gamma-1) f(N)^2-T_{f^2}^2,
\end{align*}
we see that $ f(N)T_{f^2}+T_{f^2} f(N)+\gamma f(N)^2$ is bounded from below. 
Thus it has the Friedrichs extension. 
\end{proof}
\begin{lemma}\label{dom_f2}
Let $f^2\in \mathcal K^-$. 
Then $(\one-L^\ast)\ell^2_\mathrm{fin}(\mathbb N)\subset\mathrm{D}\bigl(f(N)^2T_{f^2}\bigr)\cap\mathrm{D} (T_{f^2}f(N))$ and the operator $f(N)T_{f^2}+T_{f^2}f(N)$ is symmetric.
\end{lemma}

\begin{proof} From Lemma \ref{T_p1}, $T_{f^2}$ satisfies
$\bigl[f(N)^2, T_{f^2}\bigr]=-i\one$ 
on $(\one-L^\ast)\ell^2_\mathrm{fin}(\mathbb N)$. 
This implies that $(\one-L^\ast)\ell^2_\mathrm{fin}(\mathbb N)\subset\mathrm{D}\bigl(f(N)^2T_{f^2}\bigr)\cap\mathrm{D}(T_{f^2}f(N))$.
\end{proof}

\begin{lemma}\label{N_inv}
Let $f\in \mathcal K$. 
Then $ f(N)(\one-L^\ast)\Delta_1(f,N)^{-1}(\one-L^\ast)\ell^2_\mathrm{fin}(\mathbb N)\subset(\one-L^\ast)\ell^2_\mathrm{fin}(\mathbb N)$.
\end{lemma}

\begin{proof}
	On $\ell^2_\mathrm{fin}(\mathbb N)$, we have
\begin{align*}
&f(N)(\one-L^\ast)\Delta_1(f,N)^{-1}(\one-L^\ast)\\
&=\left(f(N)\Delta_1(f,N)^{-1}-L^\ast f(N+\one)\Delta_1(f,N)^{-1}\right)(\one-L^\ast)\\
&=\left(f(N)\Delta_1(f,N)^{-1}-L^\ast-L^\ast f(N)\Delta_1(f,N)^{-1}\right)(\one-L^\ast)\\
&=(\one-L^\ast)\left(f(N)\Delta_1(f,N)^{-1}(\one-L^\ast)-L^\ast\right).
\end{align*}
Therefore we obtain the desired result. 
\end{proof}

\begin{theorem}\label{T_f2}
Let $f^2\in \mathcal K^-$. 
Then $f(N)T_{f^2}+T_{f^2}f(N)$ and $T_f$ are time operators of $ f(N)$ with an infinite dimensional CCR-domain. 
\end{theorem}

\begin{proof} 
From Lemmas \ref{T_p1}, \ref{dom_f2} and \ref{N_inv}, we see that 
\begin{align*}
(\one-L^\ast)\Delta_1(f,N)^{-1}(\one-L^\ast)\ell^2_\mathrm{fin}(\mathbb N)
{}\subset\mathrm{D} (f(N)^2T_{f})\cap\mathrm{D} (f(N)T_{f} f(N))\cap\mathrm{D} (T_{f}f(N)^2).
\end{align*}
Therefore the symmetric operator $f(N)T_{f^2}+T_{f^2}f(N)$ satisfies
\begin{align*}
\left[f(N), f(N)T_{f^2}+T_{f^2}f(N)\right]=-i\one
\end{align*}
on $(\one-L^\ast)\Delta_1(f,N)^{-1}(\one-L^\ast)\ell^2_\mathrm{fin}(\mathbb N)$. 
By the proof of Lemma \ref{T_p1}, we can see that 
\begin{equation*}
	f(N)T_{f^2}+T_{f^2}f(N)=T_{f}
\end{equation*}
 on $(\one-L^\ast)\ell^2_\mathrm{fin}(\mathbb N)$. 
Hence $T_{f}$ is also a time operator of $f(N)$ with an infinite dimensional CCR-doamin. 
\end{proof}

%

\begin{example}
Let $ f(x)=\sqrt{x+1}$. 
Clearly, $f^2\in\mathcal K^-$. 
From Theorem \ref{T_f2}, we see that 
$f(N)T_{f^2}+T_{f^2} f(N)+ f(N)^2$ is a time operator of $ f(N)$. 
It is known that $T_{f^2}$ is bounded by \cite[Theorem 4.6]{AM08}. 
Since it has a self-adjoint extension by Proposition \ref{fr}, $ f(N)$ has a self-adjoint time operator with an infinite dimensional CCR-domain.
\end{example}

{\bf Acknowledgements:}
FH is financially supported by JSPS KAKENHI 20K20886 %
 and JSPS KAKENHI 20H01808.

\end{document}